\documentclass[a4paper,11pt]{article}

\usepackage[utf8]{inputenc}
\usepackage[margin=3cm]{geometry}
\usepackage{graphicx}
\usepackage{subcaption}
\usepackage{amsmath, amssymb, amsthm}
\usepackage{booktabs}
\usepackage{enumerate,paralist}
\usepackage[nocompress,noadjust]{cite}

\usepackage[dvipsnames]{xcolor}
\usepackage[pdfpagelabels,colorlinks,allcolors=MidnightBlue]{hyperref}

\theoremstyle{plain}

\newtheorem{theorem}{Theorem}
\newtheorem{lemma}{Lemma}

\newtheorem{corr}{Corollary}
\newtheorem{open}{Open Problem}

\theoremstyle{definition}  % switches off use of italics in following env's

\newcommand{\df}[1]{{\it #1}}

\newcommand{\tr}[1]{\ensuremath{track(#1)}}
\newcommand{\Oh}{{\ensuremath{\mathcal{O}}}}

\newcommand{\Th}{{\ensuremath{\mathcal{T}}}}

\DeclareMathOperator{\tn}{tn}
\DeclareMathOperator{\dist}{dist}

\graphicspath{{pics/}}

\newenvironment{proofof}[1]{\medskip\noindent\textbf{Proof of #1:}}{\mbox{}\hfill\qed\par\medskip}
\begin{document}

\title{Improved Bounds for Track Numbers of Planar~Graphs}

\author{
    Sergey Pupyrev \\
    spupyrev@gmail.com
}

\date{}
\maketitle

\begin{abstract}
A \emph{track layout} of a graph consists of a vertex coloring and a total order of each color class, such that 
no two edges cross between any two color classes. The \emph{track number} of a graph is the minimum number 
of colors required by a track layout of the graph.

This paper improves lower and upper bounds on the track number of several families of planar graphs. We prove that every planar graph
has track number at most $225$ and every planar $3$-tree has track number at most $25$. Then we show that there exist outerplanar graphs
whose track number is $5$, which leads to the best known lower bound of $8$ for planar graphs.
Finally, we investigate leveled planar graphs and tighten bounds on the track number of weakly leveled graphs, Halin graphs, and X-trees.
\end{abstract}

\section{Introduction}

A \df{track layout} of a graph is a partition of its vertices into sequences, called \df{tracks}, such that the vertices in each sequence form an independent set and no two edges between a pair of tracks cross each other.
The \df{track number} of a graph is the minimum number of tracks in a track layout.
Track layouts are formally introduced by Dujmovi{\'c}, Morin, and Wood~\cite{DMW05}, although similar concepts are implicitly studied in several earlier works~\cite{HR92,GLM05,FLW06}. An original motivation
for studying track layouts is their connection with the existence of low-volume three-dimensional graph drawings: 
A graph with $n$ vertices has a three-dimensional straight-line drawing in a grid of size $\Oh(1) \times \Oh(1) \times \Oh(n)$ if and only if it has track number $\Oh(1)$~\cite{DPW04,DMW05}.

Track layouts are closely related to other models of linear graph layouts, specifically, stack and queue layouts. A \df{stack}
layout of a graph consists of a linear order on the vertices and a partition of its edges so that no edges in a single part cross;
that is, there are no edges $(u, v)$ and $(x, y)$ in the same part with $u < x < v < y$. A \df{queue} layout is defined
similarly except no two edges in the same part \df{nest}; that is, 
there are no edges $(u, v)$ and $(x, y)$ in a part with $u < x < y < v$.
The minimum number of parts needed in a stack (queue) layout of a graph
is called its \df{stack number} (\df{queue number})~\cite{Oll73,HR92}.
A major result in the field is that track number is tied to queue number in a sense that one is bounded by a function of the other~\cite{DPW04}.
In particular, every $t$-track graph has a $(t-1)$-queue layout, and every $q$-queue graph has track number at most $4q \cdot 4q^{(2q-1)(4q-1)}$. The relationship between stack and track layouts is not that prominent but it is known that stack number is bounded by track number for bipartite graphs~\cite{DPW04}; whether the reverse is true is an open question.
Therefore, a study on track layouts may shed light on the relationship between the linear graph layouts.

In this paper we investigate lower and upper bounds on the track number of various families of planar graphs. Table~\ref{table:tracks} summarizes new and existing bounds described in the literature.
A recent breakthrough result by Dujmovi{\'c}, Joret, Micek, Morin, Ueckerdt, and Wood~\cite{DJMMUW19} implies that all planar graphs have bounded track number (in fact, the result extends to every proper minor-closed class of graphs). Their analysis leads to a very large constant, as it is based on a fairly generic technique. We provide an alternative proof resulting in the upper bound of $225$ for the track number of planar graphs (Section~\ref{sect:general_upper}). 
An important ingredient of our construction is an improved upper bound of $25$ for the track number of planar $3$-trees (Section~\ref{sect:bounded_treewidth}). 
For the lower bounds, we find an outerplanar graph that requires $5$ tracks, which is worst-case optimal. This resolves an open question posed in~\cite{DPW04} and provides the best lower bound of $8$ on the track number of general planar graphs (Section~\ref{sect:general_lower}).

Finally in Section~\ref{sect:other}, we study track layouts of (weakly) leveled planar graphs, which are the graphs with planar leveled drawings having no dummy vertices. This is a well studied family of planar graphs, as it is related to 
layered graph drawing and $1$-queue layouts~\cite{HR92}; refer to Section~\ref{sect:other} for a definition.
We prove that the existing upper bound of $6$ is worst-case optimal for the class of graphs, while certain subfamilies (for example, X-trees) admit a layout on $5$ tracks. The results close the gaps between upper and lower bounds on the track numbers for the subclasses of planar graphs. Our lower bounds in the section rely on computational experiments using
a SAT formulation of the track layout problem (Section~\ref{sect:sat}). We conclude the paper in Section~\ref{sect:open} with possible future research directions and open problems.

\begin{table}[!t]
	\centering
	\caption{Track numbers of various families of planar graphs}
	\label{table:tracks}
	\medskip
	\resizebox{\columnwidth}{!}{%	
		\begin{tabular}{p{3.0cm}|c@{\hspace{.9em}}c@{\hspace{.9em}}c@{\hspace{.9em}}c@{\hspace{.9em}}|c@{\hspace{.9em}}c@{\hspace{.9em}}c@{\hspace{.9em}}c@{\hspace{.9em}}}
			\toprule
			& \multicolumn{4}{c|}{Upper bound} & \multicolumn{4}{c}{Lower bound} \\
			\cmidrule(l{1ex}r{1ex}){2-5}
			\cmidrule(l{1ex}r{1ex}){6-9}
			Graph class & 
			\multicolumn{1}{c}{Old} & \multicolumn{1}{c}{Ref.} & \multicolumn{1}{c}{New} & \multicolumn{1}{c|}{Ref.} & 
			\multicolumn{1}{c}{Old} & \multicolumn{1}{c}{Ref.} & \multicolumn{1}{c}{New} & \multicolumn{1}{c}{Ref.}\\
			\midrule
			tree  
			& $3$ & \cite{FLW06} & & & $3$ & \cite{FLW06} & & \\
			
			outerplanar
			& $5$ & \cite{DPW04} & & & $4$ & \cite{DPW04} & $\mathbf{5}$ & [Thm.~\ref{thm:outer-track-5}] \\
			
			series-parallel
			& $15$ & \cite{GLM05} & & & $6$ & \cite{DMW05} & & \\
			
			planar 3-tree
			& $4,\!000$ & \cite{ABGKP18} & $\mathbf{25}$ & [Thm.~\ref{thm:p3t-25}] & $6$ & \cite{DMW05} & $\mathbf{8}$ & [Cor.~\ref{cor:planar-track-8}] \\
			
			planar
			& $461,\!184,\!080$ & \cite{DJMMUW19} & $\mathbf{225}$ & [Thm.~\ref{thm:planar-track}] & $7$ & \cite{DPW04} & $\mathbf{8}$ & [Cor.~\ref{cor:planar-track-8}] \\
			
			X-tree
			& $6$ & \cite{BDDEW15} & $\mathbf{5}$ & [Thm.~\ref{thm:xtree-upper}] & $3$ & \cite{GLM05} & $\mathbf{5}$ & [Thm.~\ref{thm:other}] \\
			
			Halin
			& $6$ & \cite{BDDEW15} & & & $3$ & \cite{DPW04} & $\mathbf{5}$ & [Thm.~\ref{thm:other}] \\
			
			weakly leveled
			& $6$ & \cite{BDDEW15} & & & $3$ & \cite{BDDEW15} & $\mathbf{6}$ & [Thm.~\ref{thm:other}] \\
			
			\bottomrule
		\end{tabular}
	}%
\end{table}

\section{Preliminaries}
\label{sect:prel}

In this section we introduce necessary definitions and recall some known results about track layouts.
Throughout the paper, $G = \big(V(G), E(G)\big)$ is a simple undirected graph with $n = |V(G)|$ vertices and $m = |E(G)|$ edges.

\paragraph{Track Layouts}

Let $\{V_i~:~1\le i \le t\}$ be a partition of $V$
such that for every edge $(u, v) \in E$, if $u \in V_i$ and $v \in V_j$, then $i \neq j$.
Suppose that $<_i$ is a total order of the vertices in $V_i$. Then the ordered set $(V_i, <_i)$ is
called a \df{track} and the partition is called a \df{t-track assignment} of $G$.
An \df{X-crossing} in a track assignment consists of two edges, $(u, v)$ and $(x, y)$, such that
$u$ and $x$ are on the same track $V_i$ with $u <_i x$, and $v$ and $y$ are on a different track $V_j$ with $y <_j v$. 
A \df{$t$-track layout} of graph $G$, denoted $\Th(G)$, is a $t$-track assignment with no X-crossings, and
\df{track number}, $\tn(G)$, is the minimum $t$ such that $G$ has a $t$-track layout. In particular, 
$1$-track graphs have no edges, and $2$-track graphs are the forests of caterpillars.

Some authors consider a relaxed definition of the concept, called \df{improper track layouts}, in which edges between consecutive vertices in a track are 
allowed~\cite{DiM03,FLW06}. It can be easily seen that the tracks of such a layout can be doubled to obtain a proper track layout~\cite{DMW05}. Thus, every graph with improper track number $t$ has (proper) track number at most $2t$. In Section~\ref{sect:other}, we show that the upper bound can be smaller than $2t$ for some subclasses of graphs. In this paper we study only proper track layouts.

A basic result on track layouts is a ``wrapping'' lemma, which is due to Dujmovi{\'c}~et~al.~\cite{DPW04} (which in turn is based on the ideas of Felsner et al.~\cite{FLW06}).
Consider a track assignment whose index set is $D$-dimensional. That is, 
let $\{V_{i_1, \dots, i_D} : 1 \le i_1 \text{ and } 1 \le i_d \le b_d \text{ for } 2 \le d \le D\}$ be a track layout of a graph $G$. 
Define the \df{partial span} of an edge $(u, v) \in E$ with 
$u \in V_{i_1, \dots, i_D}$ and $v \in V_{j_1, \dots, j_D}$ to
be $|i_1 - j_1|$. The following lemma describes how to modify the track layout of $G$ with possibly many tracks into a layout whose track number is bounded by a function of $b_2, \dots, b_D$ and the maximum partial span.

\begin{lemma}[Dujmovi{\'c} et al.~\cite{DPW04}]
	\label{lm:wrap}	
	Let $\{V_{i_1, \dots, i_D} : 1 \le i_1 \text{ and } 1 \le i_d \le b_d \text{ for }\\ 2 \le d \le D\}$ 
	be a track layout of a graph $G$ with maximum
	partial span $s \ge 1$. Then $$\tn(G) \le (2s+1) \cdot \prod_{2 \le d \le D} b_d$$
\end{lemma}	

We stress that Lemma~\ref{lm:wrap} can be applied for track layouts whose index set is one-dimensional.
In that case, $\tn(G) \le 2s+1$, where $s$ is the maximum span of an edge for the single dimension.

\paragraph{Treewidth and Tree-Partitions}

A \df{tree-decomposition} of a graph $G$ represents the vertices of $G$ as subtrees of a tree, in such a way that the vertices
are adjacent if and only if the corresponding subtrees intersect. The \df{width} of a tree-decomposition
is one less than the maximum size of a set of mutually intersecting subtrees, and the \df{treewidth}
of $G$ is the minimum width of a tree-decomposition of $G$.
For a fixed integer $k \ge 1$, a \df{k-tree} is a maximal graph of treewidth $k$, such that no more edges can be added 
without increasing its treewidth.
Alternatively, a $k$-tree is defined recursively as follows:
A $k$-clique is a $k$-tree, and the graph obtained from a $k$-tree by adding a new vertex 
adjacent to every vertex of a $k$-clique is also a $k$-tree.
A subgraph of a $k$-tree is called a \df{partial $k$-tree}.
In our proofs we do not directly use a tree-decomposition of a graph but utilize a related concept, 
a tree-partition, which is defined next.

Given a graph $G$, a \df{tree-partition} of $G$ is a pair $\big(T, \{T_x : x\in V(T)\}\big)$ consisting of a 
tree $T$ and a partition of $V$ into sets $\{T_x : x\in V(T)\}$,
such that for every edge $(u, v) \in E$ one of the following holds:
(i)~$u, v \in T_x$ for some $x \in V(T)$, or (ii)~there is an edge $(x, y)$ of $T$ with $u \in T_x$ and $v \in T_y$.
The vertices of $T$ are called the \df{nodes} and the sets $T_x, x \in V(T)$ are called the \df{bags} of the tree-partition.

The following well-known result provides a tree-partition of a $k$-tree.

\begin{lemma}[Dujmovi{\'c} et al.~\cite{DMW05}]
	\label{lm:tpart}	
	There exists a rooted tree-partition $\big(T, \{T_x : x\in V(T)\}\big)$ of a $k$-tree $G$ such that
	\begin{itemize}
		\item for every node $x$ of $T$, the subgraph of $G$ induced by the vertices of $T_x$ is a
		connected partial $(k-1)$-tree;
		\item for every non-root node $x$ of $T$, if $y$ is a parent node of $x$ in $T$, then the set of
		vertices in $T_y$ having a neighbor in $T_x$ forms a clique of size $k$ in $G$.
	\end{itemize}	
\end{lemma}	

Let us describe how one can obtain a tree-partition of a $k$-tree $G$ as in Lemma~\ref{lm:tpart}.
Fix an arbitrary vertex $r \in V$ and perform a breadth-first search in $G$ starting from $r$. For every $d \ge 0$ and every connected component induced by the vertices of $G$ at distance $d$ from the root, create a node of $T$ associated with a bag containing the vertices of the component. Two nodes are adjacent if the vertices of the corresponding bags are joined by at least one edge of $G$.
Dujmovi{\'c} et al.~\cite{DMW05} show that the constructed graph $T$ is indeed a tree, and that the vertices of each bag
form a connected subgraph of a $(k-1)$-tree.

\paragraph{Layerings and $H$-Partitions}

A generalization of a tree-partition is the notion of an $H$-partition.
An \df{$H$-partition} of a graph $G$ is a partition of $V(G)$ into disjoint \df{bags}, $\{A_x : x\in V(H)\}$ indexed by
the vertices of $H$,
such that for every edge $(u, v) \in E(G)$ one of the following holds:
(i)~$u, v \in A_x$ for some $x \in V(H)$, or (ii)~there is an edge $(x, y) \in E(H)$ with $u \in A_x$ and $v \in A_y$.
In the former case, $(u, v)$ is an \df{intra-bag} edge and in the latter case, it is an \df{inter-bag} edge.

%A \df{partition} of a graph $G$ is a set $\Ph$ of non-empty sets of vertices in $G$ such that
%every vertex belongs to exactly one element of $\Ph$. The \df{quotient} of $\Ph$ is the graph, 
%denoted by $G/\Ph$, with vertex set $\Ph$ where distinct parts $A, B \in \Ph$
%are adjacent in $G/\Ph$ if and only if some vertex in $A$ is adjacent to some vertex in $B$.
%A partition $\Ph$ of a graph $G$ is called an $H$-partition if $H$ is a graph that contains a spanning
%subgraph isomorphic to the quotient $G/\Ph$.
%
A layering of a graph $G = (V, E)$ is an ordered partition $(V_0, V_1, \dots)$ of $V$ such that for
every edge $(v, w) \in E$, if $v \in V_i$ and $w \in V_j$, then $|i - j| \le 1$.
%If $i = j$ then $(v, w)$ is an \df{intra-layer} edge. If $|i - j| = 1$ then $(v, w)$ is an \df{inter-layer} edge.
If $r$ is a vertex in a connected graph $G$ and $V_i = \{v \in V~|~\dist_G(r, v) = i\}$ for all $i > 0$, then
$(V_0, V_1, \dots)$ is called a \df{BFS-layering} of $G$. 
The \df{layered width} of an $H$-partition of a graph $G$ is the minimum integer $\ell$ such that
for some layering $(V_0, V_1, \dots)$ of $G$, we have $|A_x \cap V_i| \le \ell$ for every bag $A_x$ of the partition and every integer $i \ge 0$.

The following lemma is a key ingredient for our proof of the upper bound on the track number of planar graphs.

\begin{lemma}[Dujmovi{\'c} et al.~\cite{DJMMUW19}]
	\label{lm:hpart}	
	Every planar graph $G$ has an $H$-partition of layered width~$3$ such that $H$ is planar and has treewidth at most $3$.
	Moreover, there is such a partition for every BFS-layering of $G$.
\end{lemma}	

\section{Planar Graphs of Bounded Treewidth}
\label{sect:bounded_treewidth}

In this section, we study track numbers of planar graphs of bounded treewidth. 
Our primary goal is improving the existing upper bound for planar 3-trees.

\begin{figure}[b]
	\centering
	\includegraphics[page=1,height=2cm]{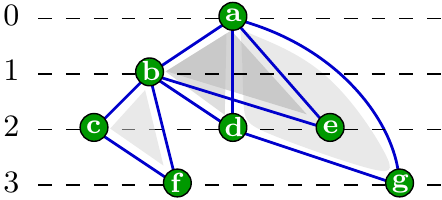}
	\caption{A $1$-clique-colorable $4$-track layout of a set of $3$-cliques.
		A nice order $\prec$ is defined as follows: 
		$\left\langle b, c, f \right\rangle \prec 
		 \left\langle a, b, d \right\rangle \prec
		 \left\langle a, d, g \right\rangle \prec 
		 \left\langle a, b, e \right\rangle$.}
	\label{fig:nicely}
\end{figure}

The following notion is introduced by Dujmovi{\'c} et al.~\cite{DMW05}; it is also implicit in the work of Di Giacomo et al.~\cite{GLM05}.
Let $\{V_i~:~1\le i \le t\}$ be a $t$-track layout of a $k$-tree, $G_k$. 
We say that a clique $C_1$ of $G_k$ \df{precedes} a clique $C_2$ of $G_k$ with respect to the track layout if
for every $1 \le i \le t$ and for all $u \in V_i \cap C_1$ and $w \in V_i \cap C_2$, it holds that $u \le_{i} w$; 
we denote the relation by $C_1 \prec C_2$.
Let $S = \{C_1, \dots, C_{|S|}\}$ be a set of maximal cliques of $G_k$. We say that $S$ is \df{nicely ordered} if $\prec$ is a total
order on $S$, that is, $C_i \prec C_j$ for all $1 \le i < j \le |S|$. Finally, for a track layout of $G_k$, we say that a given set of maximal cliques
is \df{$c$-colorable} if the set can be partitioned into $c$ nicely ordered subsets.
Such a track layout is called \df{$c$-clique-colorable}; see Figure~\ref{fig:nicely} for an illustration.

\pltopsep=0.5em
\plparsep=0.5em
\begin{lemma}
	\label{lm:base}
	The following holds:
	\begin{compactenum}[\bf(a)]
		\item \label{lm:base_path} 
		Every path admits a $1$-clique-colorable $2$-track layout.
		
		\item \label{lm:base_tree}
		Every tree admits a $2$-clique-colorable $3$-track layout.
		
		\item \label{lm:base_outer}
		Every outerplanar graph admits a $2$-clique-colorable $5$-track layout.
	\end{compactenum}	
\end{lemma}	

\begin{proof}
	Claim (\ref{lm:base_path}) of the lemma is straightforward, as every $2$-track layout of a path is $1$-clique-colorable. 
	Now we prove~(\ref{lm:base_tree}).
	Consider a plane drawing of a tree, $T$, such that the vertices having the same distance from the root are drawn on the same horizontal line;
	see Figure~\ref{fig:cc1}. Let $x(v)$ and $y(v)$ denote $x$ and $y$ coordinates of a vertex $v \in V(T)$.
	The drawing corresponds to a track layout with maximum span $1$, which by Lemma~\ref{lm:wrap} can be converted into a $3$-track layout by
	assigning $\tr{v} = y(v)~(\bmod~3)$ for every vertex $v$.
	The maximal cliques in the graph are edges of $T$. We partition the edges into $S_1 = \{(u, v) \in E(T) : \tr{u} = 1 \text{~or~} \tr{v} = 1\}$ and 
	$S_2 = E(T) \setminus S_1$. It is easy to verify that the two sets of edges are nicely ordered with respect to the $3$-track layout.
	
	In order to prove (\ref{lm:base_outer}), we utilize a $5$-track layout of an outerplanar graph suggested by Dujmovi{\'c} et al.~\cite{DPW04}; see also~\cite{ABGKP18}. 
	They prove that every maximal outerplanar graph, $G$, has a straight-line outerplanar drawing in which vertex coordinates are integers, and the absolute value of the difference of the $y$-coordinates of the endvertices of each edge of $G$ is either one or two; see Figure~\ref{fig:cc2}. That is, $1 \le |y(v) - y(u)| \le 2$ for all
	$(u, v) \in E$.	Such a drawing defines a track layout in which vertices with the same $y$-coordinates
	form a track and the ordering of the vertices within each track is implied by the drawing.
	The layout may have many tracks but maximum edge span is $2$.
	By Lemma~\ref{lm:wrap}, the layout is wrapped onto $5$ tracks by assigning $\tr{v} = y(v)~(\bmod~5)$ and ordering vertices within a track lexicographically by $\big(\lfloor y/5 \rfloor, x\big)$.
	
	The maximal cliques in $G$ are triangular faces. Denote a face containing vertices $u, v, w$ by $\left\langle u, v, w \right\rangle $ and the set of all faces in $G$ by $F$. We partition $F$ into two sets
	$S_1 = \{\left\langle u, v, w \right\rangle \in F : \tr{u} = 2 \text{~or~} \tr{v} = 2 \text{~or~} \tr{w} = 2\}$ and $S_2 = F \setminus S_1$; see Figure~\ref{fig:cc2} in which members of $S_2$ are shaded. All the faces of $S_1$ contain a vertex in track $2$; choose the leftmost (having smallest $x$-coordinate) such vertex
	$v^*$ with $\tr{v^*} = 2$. The faces containing $v^*$ can be nicely ordered with respect to the track layout, as the drawing is planar.
	Removing $v^*$ from the layout and applying the same argument for the remaining faces, yields a nice order of $S_1$.	
	An analogous procedure can be utilized to construct a nice order of $S_2$, since all those faces contain a vertex in track $0$ and a vertex in track $4$.
\end{proof}

\begin{figure}[t]
	\centering
	\begin{subfigure}[b]{\linewidth}
		\centering
		\includegraphics[page=3,height=2.7cm]{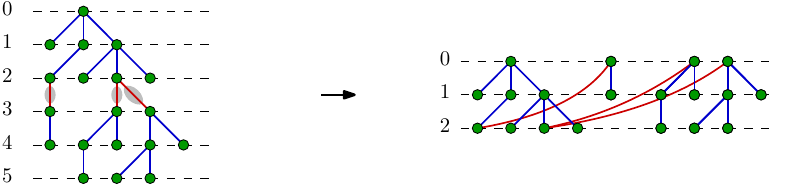}
		\caption{A track layout of a tree wrapped onto a $2$-clique-colorable $3$-track layout. 
			Edges ($2$-cliques) are partitioned into two nice orders, blue and shaded red.\\}
		\label{fig:cc1}
	\end{subfigure}
	\begin{subfigure}[b]{\linewidth}
		\centering
		\includegraphics[page=5,height=5cm]{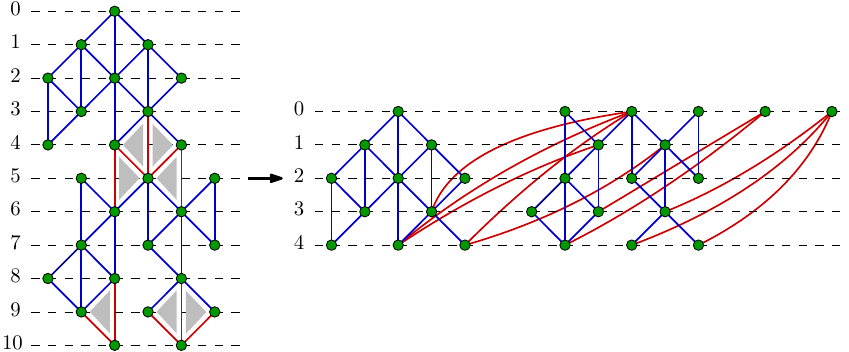}
		\caption{A track layout of an outerplanar graph wrapped onto a $2$-clique-colorable $5$-track layout. 
			Faces ($3$-cliques) are partitioned into $2$ nice orders; faces of one part are shaded.}
		\label{fig:cc2}
	\end{subfigure}
	\caption{An illustration for Lemma~\ref{lm:base}}
	\label{fig:cc}
\end{figure}

Notice that the claims of Lemma~\ref{lm:base} are tight in terms of clique-colorability: There is no constant $t$ such that
every tree admits a $1$-clique-colorable $t$-track layout.
The next lemma shows how to use clique-colorable track layouts for graphs of bounded treewidth.

\begin{lemma}
	\label{lm:tc}
	Assume that every $k$-tree admits a $c$-clique-colorable $t$-track layout. 
	Then every $(k+1)$-tree admits a layout on $t \cdot (2c+1)$ tracks.
\end{lemma}

\begin{proof}
	Let $\big(T, \{T_x : x \in V(T)\}\big)$ be a tree-partition of a $(k+1)$-tree, $G_{k+1}$, given by Lemma~\ref{lm:tpart}.
	In order to construct a desired track layout of $G_{k+1}$, consider a vertex $v \in V(G_{k+1})$ that belongs to a bag $x \in V(T)$ in the tree-partition.
	The track of $v$ is defined by two indices, $(i_v, j_v)$, where $i_v$ is derived from a track layout of $T_x$ and $j_v$ is derived from a 
	certain track layout of the tree $T$. Next we define $i_v$ and $j_v$.
	
	\begin{itemize}
		\item By Lemma~\ref{lm:tpart}, vertices of every bag of the tree-partition form a connected partial $k$-tree, 
		which is a subgraph of a $k$-tree~\cite{Bod98}.
		Thus, $T_x$ admits a $c$-clique-colorable $t$-track layout,
		which we denote by $\Th(T_x)$. The first index, $i_v$, is the track of $v$ in $\Th(T_x)$.
		Clearly the index ranges from $1$ to $t$.
		
		\item Consider a parent node, $y \in V(T)$, of node $x$ in the tree-partition. By Lemma~\ref{lm:tpart}, vertices of $T_y$ adjacent to a vertex of $T_x$ form a maximal clique, which we call the \df{parent clique} of $x$. The vertices of $T_y$ form a 
		partial $k$-tree, which by the assumption of the lemma admits a $c$-clique-colorable $t$-track layout, $\Th(T_y)$. Thus, the
		parent clique of $x$ has an assigned color, $c(x)$, ranging from $1$ to $c$, and
		cliques with the same color are nicely ordered in $\Th(T_y)$.
		
		We layout $T$ on (possibly many) tracks such that, for every parent node $y$, its child
		nodes are on $c$ consecutive tracks. Formally, if $x$ is a child of $y$ and the corresponding parent clique of $x$ has color $1 \le c(x) \le c$, then
		$\tr{x} = \tr{y} + c(x)$. The order of child nodes having the same color follows the nice order of the corresponding parent cliques.
		The constructed track layout of $T$ is denoted $\Th(T)$.
		
		We assign $j_v$ to be the track of node $x$ in $\Th(T)$. The index can be as large as $\Omega(|V(T)|)$ but the maximum
		span of an edge of $E(T)$ is $c$.
	\end{itemize}
	
	Now we define the order of the vertices of $G_{k+1}$ within the same track. If two vertices, $u$ and $v$, belong to the same bag
	$x \in V(T)$ in the tree-partition, their relative order is inherited from track layout $\Th(T_x)$. If the vertices are
	in different bags, that is, $u \in T_x, v \in T_y$ for some $x \in V(T), y \in V(T)$, then the order is dictated by the order of nodes $x$ and $y$ in track layout $\Th(T)$.
	
	Next we show that the constructed track assignment has no X-crossings. Intra-bag edges do not form X-crossings by the assumption of the lemma.
	Consider two inter-bag edges, $(u_1, v_1) \in E(G_{k+1})$ and $(u_2, v_2) \in E(G_{k+1})$. Since there are no crossings in $\Th(T)$, the
	inter-bag edges mapped to edges of $T$ without a common parent are not in an X-crossing. Thus we may assume that $u_1 \in T_p, u_2 \in T_p, v_1 \in T_x, v_2 \in T_y$
	for some nodes $p, x, y \in V(T)$ such that $p$ is a parent of $x$ and $y$.
	If parent cliques of $x$ and $y$ are of different colors, then by construction of $\Th(T)$, vertices $v_1$ and $v_2$ are in different tracks.
	If parent cliques of $x$ and $y$ are of the same color, then the order between $v_1$ and $v_2$ is consistent with the order between 
	$u_1$ and $u_2$, since the cliques are nicely ordered in $\Th(T_p)$. Therefore, an X-crossing between $(u_1, v_1)$ and $(u_2, v_2)$ is impossible.
	
	Finally, we apply Lemma~\ref{lm:wrap} for the constructed two-dimensional track layout with maximum partial span $c$ to get the
	desired claim.
\end{proof}

By combining Lemma~\ref{lm:tc} with Lemma~\ref{lm:base}, we can get an improved upper bound for the track number of planar $3$-trees. \begin{theorem}
	\label{thm:p3t-25}
	The track number of a planar $3$-tree is at most $25$.
\end{theorem}

\begin{proof}
	First we build a tree-partition of a planar $3$-tree $G$ as in Lemma~\ref{lm:tpart}.
	To this end, as mentioned in Section~\ref{sect:prel}, pick a root $r \in V$ and perform a 
	breadth-first search in $G$ starting from $r$. Bags of the tree-partition are formed by connected 
	components induced by the vertices of $G$ at the same distance from the root.
	
	Since $G$ is planar, vertices at the same distance from the root induce an outerplanar graph~\cite{ABGKP18}.
	Hence, the vertices of each bag of the tree-partition form a connected outerplanar graph, which by Lemma~\ref{lm:base}
	admits a $2$-clique-colorable $5$-track layout.
	Applying the arguments of Lemma~\ref{lm:tc} to a planar $3$-tree $G$ yields the claim of the theorem.
\end{proof}	

Note that a similar construction provides an upper bound of $15$ for planar $2$-trees (series-parallel graphs), as
every bag of a corresponding tree-partition induce a tree, which admits a $2$-clique-colorable $3$-track layout.
The same upper bound is already known~\cite{GLM05}.

\section{General Planar Graphs}

In this section we investigate upper and lower bounds on the track number of general planar graphs.

\subsection{An Upper Bound}
\label{sect:general_upper}

Recently Dujmovi{\'c} et al.~\cite{DJMMUW19} used $H$-partitions of bounded layered width to prove that the
queue number of a planar graph is a constant. By analogy with their result, we show that the track
number of a graph is bounded by a function of the track number of $H$ and the layered width.

\begin{lemma}
	\label{lm:main}
	If a graph $G$ has a layered $H$-partition of layered width $\ell$, then $G$ has track number
	at most $3 \ell \cdot \tn(H)$.
\end{lemma}

\begin{proof}
	Assume $G$ has a layered $H$-partition of width $\ell$, and suppose $\Th(H)$ is a layout of $H$ 
	on $\tn(H)$ tracks. To define a track assignment of $G$, $\Th(G)$, consider a vertex $v \in V(G)$ that belongs to a bag $x \in V(H)$.
	The track of $v$ is defined by three indices, $(i_v, j_v, d_v)$.
	\begin{itemize}
		\item The first index, $i_v$, is the track of $x$ in track layout $\Th(H)$; it ranges from $1$ to $\tn(H)$. 
		\item The bag $x \in V(H)$ contains at most $\ell$ vertices in every layer. Label these
		vertices arbitrarily from $1$ to $\ell$, and assign the second index, $j_v$, to the label.
		Thus, $1 \le j_v \le \ell$.
		\item The last index, $d_v$, represents the layer of $v$ in the given layered $H$-partition of $G$.
		Clearly, $d_v \ge 1$ and $d_v$ is at most the number of layers in $G$, which
		can be as large as $\Omega(|V(G)|)$.
	\end{itemize}		
	%	We first define a track assignment of $G$, $\Th(G)$;  
	
	In order to complete the track assignment, we define the order of vertices in the same track.
	Notice that vertices of every bag, $A_x$ for some $x \in V(H)$, are on different
	tracks defined by the second and the third indices of the track assignment. Therefore, only
	the vertices of $G$ corresponding to different bags can belong to the same track. For those
	vertices, the order is inherited from the given track layout $\Th(H)$. That is, $v < u$ in $\Th(G)$
	with $v \in A_x, u \in A_y$ if and only if $x < y$ in $\Th(H)$; see Figure~\ref{fig:no2}.
	
	Now we verify that $\Th(G)$ is a valid track layout, that is, it contains no X-crossings. For a 
	contradiction suppose that $(u_1, v_1) \in E(G)$ and $(u_2, v_2) \in E(G)$ form an X-crossing, that is,
	$\tr{u_1} = \tr{u_2}, \tr{v_1} = \tr{v_2}$ and $u_1 < u_2, v_2 < v_1$. Since all vertices of a bag are
	on different tracks, it follows that $u_1$ and $u_2$ belong to different bags and that 
	$v_1$ and $v_2$ belong to different bags. Therefore, the two edges correspond either to
	an X-crossing in $\Th(H)$ (if the two edges are intra-bag edges), or to an edge
	of $H$ with both endpoints on the same track of $\Th(H)$ (if one of the edges is an inter-bag edge).
	Both of the options violate the definition of $\Th(H)$; hence, $\Th(G)$ contains no X-crossings.
	
	Finally, observe that the partial span in $\Th(G)$ corresponding to the third dimension of 
	the track assignment, $d_v$, is at most one, as it is based on a
	layering of $G$. By Lemma~\ref{lm:wrap}, the track layout can be wrapped onto
	$3 \ell \cdot \tn(H)$ tracks; Figure~\ref{fig:no2} illustrates the process.
\end{proof}	

%Notice that the result of the lemma is tight in some sense; see Figure~\ref{fig:no2}.

\begin{figure}[!t]
	\begin{subfigure}[b]{.38\linewidth}
		\centering
		\includegraphics[page=1,height=4.2cm]{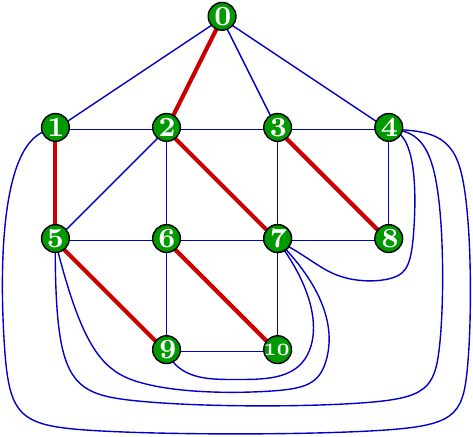}
		\caption{}
	\end{subfigure}    
	\hfill
	\begin{subfigure}[b]{.3\linewidth}
		\centering
		\includegraphics[page=2,height=4.2cm]{planar_hpart}
		\caption{}
	\end{subfigure}
	%\hfill
	\begin{subfigure}[b]{.3\linewidth}
		\centering
		\includegraphics[page=3,height=4.2cm]{planar_hpart}
		\caption{}
	\end{subfigure}
	\caption{
		(a)~A planar graph with a layered $H$-partition rooted at $v=0$ with layered width $\ell=1$;
		intra-bag edges are red.
		(b)~A $4$-track layout of $H$.
		(c) A $12$-track layout of the graph constructed using Lemma~\ref{lm:main}.}
	\label{fig:no2}
\end{figure}

Combining Lemma~\ref{lm:main}, Lemma~\ref{lm:hpart}, and Theorem~\ref{thm:p3t-25}, we get the following result.

\begin{theorem}
	\label{thm:planar-track}
	The track number of a planar graph is at most $225$.
\end{theorem}

\subsection{Lower Bounds}
\label{sect:general_lower}

Dujmovi{\'c} et al.~\cite{DPW04} show an outerplanar graph that requires $4$ tracks and 
prove that every outerplanar graph has a $5$-track layout. Our next result closes
the gap between the lower and the upper bounds answering the question posed in~\cite{DPW04}.

\begin{theorem}
	\label{thm:outer-track-5}    
	The outerplanar graph in Figure~\ref{fig:5track:ex} has track number $5$.
\end{theorem}    

\newcommand{\QC}[4]{\ensuremath{Q(#1,#2\!<\!#3,#4)}}
\newcommand{\WC}[4]{\ensuremath{W(#1\!<\!#2,#3\!<\!#4)}}

Before proving the theorem, we introduce two configurations that we use in the proof. The first configuration,
illustrated in Figure~\ref{fig:5track:Q}, is defined on four vertices forming a cycle.
If $\tr{a}=1$, $\tr{b} = \tr{c} = 2$, $\tr{d} = 3$, and $b < c$, then for
every vertex $v$ that is ``inside'' the quadrangle (that is, $\tr{v}=2$ and $b < v < c$), 
its neighbor $u$ is not on tracks $1$ and $3$. We call this a
\df{$\QC{a}{b}{c}{d}$-configuration} for vertex $v$.
The second configuration, illustrated in Figure~\ref{fig:5track:W}, is defined on four vertices, $a$, $b$, $c$, $d$, 
with $\tr{a} = \tr{b} = 1$, $\tr{c} = \tr{d} = 2$ and $a < b$, $c < d$. If there exist two
vertices $u$ and $v$ together with edges $(u, a)$, $(u, d)$, $(v, b)$, $(v, c)$, then
$\tr{u} \neq \tr{v}$. We call this a \df{$\WC{a}{b}{c}{d}$-configuration} for vertices $u$ and $v$.
We emphasize that for both configurations, the actual tracks of the vertices are irrelevant; it is
only important which vertices share tracks.

\begin{figure}[t]
	\begin{subfigure}[b]{.38\linewidth}
		\centering
		\includegraphics[page=1,width=0.8\textwidth]{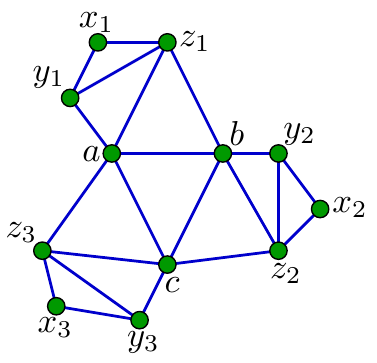}
		\caption{A graph with track number $5$}
		\label{fig:5track:ex}
	\end{subfigure}    
	\hfill
	\begin{subfigure}[b]{.30\linewidth}
		\centering
		\includegraphics[page=2,width=1.0\textwidth]{outer_track5}
		\caption{$Q$-configuration}
		\label{fig:5track:Q}
	\end{subfigure}
	\begin{subfigure}[b]{.30\linewidth}
		\centering
		\includegraphics[page=3,width=1.0\textwidth]{outer_track5}
		\caption{$W$-configuration}
		\label{fig:5track:W}
	\end{subfigure}
	\caption{An illustration for Theorem~\ref{thm:outer-track-5}}
	\label{fig:5track}
\end{figure}

\begin{proofof}{Theorem~\ref{thm:outer-track-5}}
	Assume that the graph in Figure~\ref{fig:5track:ex} has a $4$-track layout.
	Without loss of generality, we may assume that vertex $a$ is on track $1$, vertex $b$ is on track $2$, and vertex $c$
	is on track $3$. Next we consider tracks of vertices $z_1$, $z_2$, and $z_3$, and 
	distinguish four cases depending on how many of the vertices are on track $4$.
	
	\begin{itemize}
		\item None of $z_1$, $z_2$, $z_3$ are on track $4$.
		
		It is easy to see that $\tr{z_1}=3$, $\tr{z_2}=1$, and $\tr{z_3}=2$; see Figure~\ref{ot5_1}.
		Assume without loss of generality that $b < z_3$, that is, 
		vertex $b$ precedes vertex $z_3$ on track $2$. Then in order for edge $(c, z_3)$ to avoid a crossing with edge
		$(b, z_1)$, vertex $z_1$ should precede $c$ on track $3$, that is, $z_1 < c$. Then $a$ and $z_2$
		form a $\WC{z_1}{c}{b}{z_3}$-configuration, a contradiction.
		
		\item All of $z_1$, $z_2$, $z_3$ are on track $4$.
		
		Assume without loss of generality that $z_1 < z_2 < z_3$. We show that the graph in Figure~\ref{ot5_4} is not embeddable in four tracks.
		
		Observe that $z_2$'s neighbors cannot be on track $1$; otherwise, the edge from $z_2$ to the neighbor crosses one of the edges
		$(a, z_1)$ or $(a, z_3)$. Therefore, $\tr{y_2} = 3$ and $\tr{x_2} = 2$.
		Notice that $y_2 < c$, as otherwise edges $(y_2, z_2)$ and $(c, z_3)$ cross. Then $x_2$ and $b$
		form a $\WC{z_1}{z_2}{y_2}{c}$-configuration, a contradiction.
		
		\item One of $z_1$, $z_2$, $z_3$ is on track $4$; suppose $\tr{z_1} = 4$.
		
		It is easy to see that $\tr{z_2} = 1$ and $\tr{z_3} = 2$. Assume without loss of generality that $z_3 < b$; it follows that $a < z_2$, as
		otherwise $(b, z_2)$ and $(a, z_3)$ cross.
		Next we prove that the graph in Figure~\ref{ot5_2} is not embeddable in four tracks.
		We distinguish two cases depending on the track of $y_1$.
		
		First assume $\tr{y_1} = 3$. It holds that $y_1 < c$, as otherwise $(c, z_2)$ and $(y_1, a)$ cross. Now it is impossible to assign a track for $y_3$:
		If $\tr{y_3} = 1$, then $y_3$ and $a$ form a $\WC{y_1}{c}{z_3}{b}$-configuration;
		If $\tr{y_3} = 4$, then $y_3$ and $z_1$ form a $\WC{y_1}{c}{z_3}{b}$-configuration.
		
		Second assume $\tr{y_1} = 2$. It holds that $y_1 < b$, as otherwise $(b, z_2)$ and $(a, y_1)$ cross.
		If $z_3 < y_1 < b$, then $y_1$ forms a $\QC{c}{z_3}{b}{a}$-configuration.
		If $y_1 < z_3 < b$, then $z_3$ forms a $\QC{z_1}{y_1}{b}{a}$-configuration, a contradiction.
		
		\item Two of $z_1$, $z_2$, $z_3$ are on track $4$. Suppose $\tr{z_1} = \tr{z_2} = 4$; thus, $\tr{z_3} = 2$.
		
		Assume without loss of generality that $z_1 < z_2$.
		Next we prove that the graph in Figure~\ref{ot5_3} is not embeddable in four tracks. We distinguish two cases depending on
		the relative order of $b$ and $z_3$.
		
		First assume $b < z_3$. Consider the track of $y_2$. If $\tr{y_2}=1$, then $a$ and $y_2$ form a
		$\WC{z_1}{z_2}{b}{z_3}$-configuration; thus, $\tr{y_2} = 3$. To avoid a crossing between edges
		$(c, z_3)$ and $(b, y_2)$, we have $y_2 < c$. Now it is not possible to layout $x_2$.
		If $\tr{x_2} = 1$, then $a$ and $x_2$ form a $\WC{z_1}{z_2}{y_2}{c}$-configuration.
		Otherwise if $\tr{x_2} = 2$, then $b$ and $x_2$ form a $\WC{z_1}{z_2}{y_2}{c}$-configuration.
		
		Second assume $z_3 < b$. Consider the track of $y_1$. If $\tr{y_1} = 3$, then $y_1 < c$, as otherwise $(y_1, z_1)$ and $(c, z_2)$ cross.
		Now it is impossible to layout $y_3$, which
		forms and a $\WC{y_1}{c}{z_3}{b}$-configuration with vertex $a$ if $\tr{y_3} = 1$, 
		and it forms a $\WC{y_1}{c}{z_3}{b}$-configuration with vertex $z_1$ if $\tr{y_3} = 4$.
		
		Thus, $\tr{y_1} = 2$. It follows that $y_1 < b$, as otherwise $(y_1, z_1)$ and $(b, z_2)$ cross. 
		If $z_3 < y_1 < b$, then $y_1$ is in a $\QC{a}{z_3}{b}{c}$-configuration and its neighbor, $x_1$, cannot be placed.
		Otherwise if $y_1 < z_3 < b$, then $z_3$ is in a $\QC{a}{y_1}{b}{z_1}$-configuration, and $y_3$ cannot be placed.
	\end{itemize}    

	Therefore, the graph in Figure~\ref{fig:5track:ex} is not embeddable in a $4$-tracks, which concludes the proof
	of the theorem.
\end{proofof}

\begin{figure}[t]
	\centering
	\begin{subfigure}[b]{.49\linewidth}
		\includegraphics[page=4,width=\textwidth]{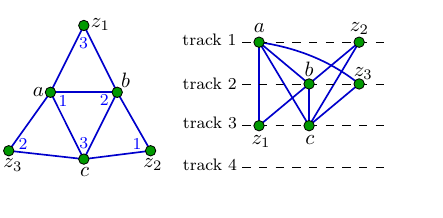}
		\caption{None of $z_1$, $z_2$, $z_3$ are on track $4$}~\\
		\label{ot5_1}
	\end{subfigure}    
\hfill
	\begin{subfigure}[b]{.49\linewidth}
		\includegraphics[page=6,width=\textwidth]{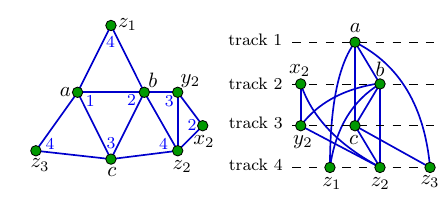}
		\caption{All of $z_1$, $z_2$, $z_3$ are on track $4$}~\\
		\label{ot5_4}
	\end{subfigure}
	\begin{subfigure}[b]{.49\linewidth}
		\includegraphics[page=5,width=\textwidth]{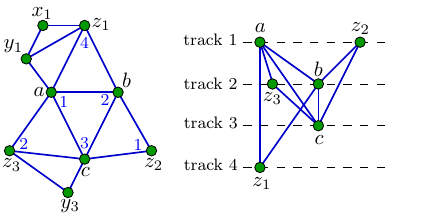}
		\caption{One of $z_1$, $z_2$, $z_3$ is on track $4$}
		\label{ot5_2}
	\end{subfigure}
\hfill
	\begin{subfigure}[b]{.49\linewidth}
		\includegraphics[page=7,width=\textwidth]{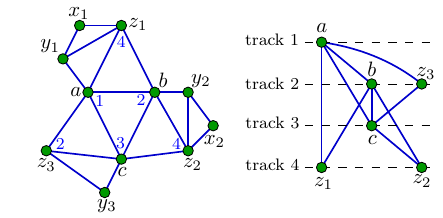}
		\caption{Two of $z_1$, $z_2$, $z_3$ are on track $4$}
		\label{ot5_3}
	\end{subfigure}
	\caption{Different cases in the proof of Theorem~\ref{thm:outer-track-5}. Track of a vertex is blue.}    
\end{figure}

Dujmovi{\'c} et al. (Lemma~22 of~\cite{DPW04}) show that, for every outerplanar graph $H$ with track number $\tn(H)$, there exists a
planar graph whose track number is $\tn(H) + 3$. We stress that the graph in their construction is a planar $3$-tree. Therefore, we have
the following improved lower bound for the track number of planar $3$-trees.

\begin{corr}
	\label{cor:planar-track-8}
	There exists a planar $3$-tree $G$ with track number $\tn(G) = 8$.
\end{corr}

\section{Other Subclasses of Planar Graphs}
\label{sect:other}

Track layouts of planar graphs are related to leveled planar graph drawings that were introduced
by Heath and Rosenberg~\cite{HR92} in the context of queue layouts.
A leveled planar drawing of a graph is a straight-line crossing-free drawing 
in the plane, such that the vertices are placed on a sequence of parallel lines (levels) and every edge joins vertices in two consecutive levels. 
A graph is \df{leveled planar} if it admits a leveled planar drawing. Bannister et al.~\cite{BDDEW15} characterize the class
of graphs by showing that a graph is leveled planar if and only if it is bipartite and admits a 3-track layout. 
In a relaxed definition of leveled drawings, edges between consecutive vertices on the same level are allowed; this leads to \df{weakly leveled planar} graphs.
For graphs that have a weakly leveled planar drawing, Bannister et al.~\cite{BDDEW15} show the upper bound of $6$ for the track number, while
leaving the question of the lower bound open. We answer the question by providing an example of a weakly leveled planar graph whose track number is $6$.

Certain families of planar graphs are known to admit weakly leveled planar graphs; for example, Halin graphs (an embedded tree with no vertices of degree $2$
whose leaves are connected by a cycle) and X-trees (a complete binary tree with extra edges connecting vertices of the same level).
Although track numbers of the graphs have been investigated in several earlier works~\cite{DiM03,GLM05,BDDEW15}, the gaps between lower and upper
bounds remain open. In the following, we present a Halin graph and an X-tree that require $5$ tracks. In addition we provide an
algorithm that constructs a $5$-track layout for every X-tree, thus, closing the gap between lower and upper bounds of the track number of such graphs.

Our lower bound examples in the section rely on computational experiments. To this end, we propose a SAT formulation of the track layout problem, 
and share the source code of our implementation~\cite{bob}. The formulation is simple-to-implement but efficient enough to find optimal 
track layouts of graphs with up to a few hundred of vertices in a reasonable amount of time.

\subsection{An Upper Bound for X-trees}

%\subsubsection{X-trees}
An \df{X-tree} is a complete binary tree with extra edges connecting vertices of the same level. Formally, 
if $v_1, v_2, \dots, v_{2^d}$ are the vertices of level $d \ge 0$ in the tree in the left-to-right order, 
then the extra edges are $(v_i, v_{i+1})$ for $1 \le i < 2^{d}$; see Figure~\ref{fig:xtree0}. 
Since X-trees admit a weakly leveled planar drawing, their track number is at most $6$~\cite{DiM03,GLM05,BDDEW15}.
Next we improve the upper bound to $5$.

\begin{theorem}
	\label{thm:xtree-upper}
	Every X-tree has a $5$-track layout.
\end{theorem}

\begin{figure}[tb]
	\begin{subfigure}[b]{.4\linewidth}
		\centering
		\includegraphics[page=1,width=0.7\textwidth]{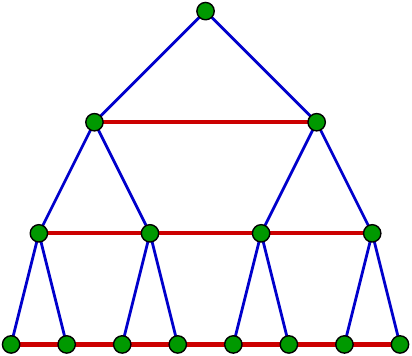}
		\caption{An X-tree with $4$ levels}
		\label{fig:xtree0}
	\end{subfigure}    
	\hfill
	\begin{subfigure}[b]{.5\linewidth}
		\centering
		\includegraphics[page=2,width=0.8\textwidth]{xtrees}
		\caption{A track layout with maximum span $2$}
		\label{fig:xtree1}
	\end{subfigure}
	\centering
	\begin{subfigure}[b]{.99\linewidth}
		\includegraphics[page=5,width=\textwidth]{xtrees}
		\caption{Adding~~vertices~~$v_{2i-1}$~~and~~$v_{2i}$ (squares) while~~maintaining~~an~~invariant
			$y(w_i) \le y(v_{2i-2}) \le y(w_i) + 3$}
		\label{fig:xtree2}
	\end{subfigure}
	\caption{An illustration for Theorem~\ref{thm:xtree-upper}: constructing a $5$-track layout for X-trees}
\end{figure}

\begin{proof}
	We call the edges between the vertices of the same level in the X-tree \df{level edges}, while the remaining edges
	are \df{tree edges}. In order to construct a $5$-track layout, we build a planar straight-line drawing
	of the graph such that every vertex $v$ is laid out on an integer grid with coordinates $x(v) \in \mathbb{N}$ and $y(v) \in \mathbb{N}$. 
	In the drawing we maintain a property that $|y(u) - y(v)| \le 2$ for every edge $(u, v)$; see Figure~\ref{fig:xtree1}.
	It is easy to see that such a drawing corresponds to a track layout with span $2$ in which the
	vertices having equal $y$-coordinates are on the same track.
	By Lemma~\ref{lm:wrap}, the layout can be wrapped onto $5$ tracks.
	
	The drawing is built inductively on the level, $d \ge 0$, of a given X-tree using a hypothesis that
	every X-tree admits a planar straight-line drawing such that:
	\begin{itemize}
		\item $|y(u) - y(v)| = 1$ for every level edge $(u, v)$ and $|y(u) - y(v)| \le 2$ for every tree edge;
		
		\item for the level-$d$ vertices in the tree, $v_1, v_2, \dots, v_{2^d}$, it holds that
		$x(v_1) < x(v_2) < \dots < x(v_{2^d})$; equivalently, the corresponding level edges form a 
		strictly $x$-monotone polyline;
		
		\item every edge of the highest level, $(u, v)$, is on the boundary of the drawing; that is, it is visible 
		from points $(x(u), +\infty)$ and $(x(v), +\infty)$.
		
	\end{itemize}	
	The basis of the induction for $d \le 1$ is trivial; see Figure~\ref{fig:xtree1}. In order to construct the drawing
	for an X-tree, $G$, for $d > 1$, we start with an inductively constructed X-tree, $G'$, of level $d-1$.
	Denote the vertices of $G$ of level $d-1$ by $w_1, \dots, w_{2^{d-1}}$ and the vertices of level $d$
	by $v_1, \dots, v_{2^d}$. Observe that $G$ can be constructed from $G'$ by iteratively adding a pair of vertices, $v_{2i-1}$ and $v_{2i}$, 
	along with three edges $(v_{2i-1}, v_{2i})$, $(v_{2i-1}, w_i)$, and $(v_{2i}, w_i)$, where $1 \le i \le 2^{d-1}$.
	Starting from a drawing of $G'$, we add the pairs of vertices while maintaining an invariant that
	$y(w_i) \le y(v_{2i-2}) \le y(w_i) + 3$ for all $2 \le i < 2^{d-1}$.
	
	Clearly, the invariant is initialized for $i=2$ by setting $x(v_1)=1, y(v_1)=y(w_1)+1$ and $x(v_2)=2, y(v_2)=y(w_1)+2$.
	For $i > 2$, we distinguish the cases based on the $y$-coordinates of the previously placed vertices. As illustrated in 
	Figure~\ref{fig:xtree2}, there are four cases:
	\begin{compactitem}
		\item[case (a) $y(v_{2i-2}) = y(w_i)~~~~~$:] assign $y(v_{2i-1}) = y(w_i) + 1$, $y(v_{2i}) = y(w_i) + 2$;
		\item[case (b) $y(v_{2i-2}) = y(w_i) + 1$:] assign $y(v_{2i-1}) = y(w_i) + 2$, $y(v_{2i}) = y(w_i) + 1$;
		\item[case (c) $y(v_{2i-2}) = y(w_i) + 2$:] assign $y(v_{2i-1}) = y(w_i) + 1$, $y(v_{2i}) = y(w_i) + 2$;
		\item[case (d) $y(v_{2i-2}) = y(w_i) + 3$:] assign $y(v_{2i-1}) = y(w_i) + 2$, $y(v_{2i}) = y(w_i) + 1$.
	\end{compactitem}	
	
	In all of the cases, assign $x(v_i) = i$ for $1\le i \le 2^d$.
	One can easily verify that the desired invariants are maintained and the resulting drawing is planar.
\end{proof}

\subsection{Lower Bounds}
\label{sect:sat}

To test whether a given graph $G = (V, E)$ admits a $t$-track layout, we formulate a Boolean Satisfiability Problem 
that has a solution if and only if $G$ has a layout on $t$ tracks. We introduce two sets of variables:

\begin{itemize}
	\item a variable $\phi_q(v)$ for every vertex $v \in V$ and every track $1 \le q \le t$ indicating whether the vertex
	belongs to the track;
	\item a variable $\sigma(v, u)$ for every pair of vertices $v \in V, u \in V$ indicating whether $v$ precedes $u$ in the 
	order, that is, $v < u$ for some track of the layout.
\end{itemize}	

\noindent
Every vertex is assigned to one track, which is ensured by the track assignment rules:
$$
\phi_1(v) \vee \phi_2(v) \vee \dots \vee \phi_t(v)~~\forall v \in V \text{~~~and~~~} \neg \phi_i(v) \vee \neg \phi_j(v)~~\forall~v\in V, 1\le i < j \le t
$$

\noindent
To guarantee a valid track assignment, we forbid adjacent vertices on the same track:
$$
\neg \phi_i(v) \vee \neg \phi_i(u)~~\forall~(u, v) \in E, 1 \le i \le t
$$

\noindent
For the relative encoding of vertices, we ensure asymmetry and transitivity:
$$
\sigma(v, u) \leftrightarrow \neg \sigma(u, v)~~\forall \text{ distinct } u, v \in V
$$
$$
\sigma(v, u) \wedge \sigma(u, w) \rightarrow \sigma(v, w)~~\forall \text{ distinct } u, v, w \in V
$$

\noindent
To forbid X-crossings among edges, recall that an X-crossing between edges $(u_1, v_1) \in E, (u_2, v_2) \in E$ occurs
when $\tr{u_1} = \tr{u_2}, \tr{v_1} = \tr{v_2}$ and $u_1 < u_2, v_1 > v_2$. This can be expressed by the following rules:
$$
\phi_i(u_1) \wedge \phi_i(u_2) \wedge \phi_j(v_1) \wedge \phi_j(v_2) \rightarrow 
\big(\sigma(u_1, u_2) \wedge \sigma(v_1, v_2)\big) \vee \big(\sigma(u_2, u_1) \wedge \sigma(v_2, v_1)\big)
$$
$$
\forall~(u_1, v_1) \in E, (u_2, v_2) \in E \text{ such that } u_1 \neq u_2, v_1 \neq v_2 \text{ and } \forall~1 \le i, j \le t
$$

The resulting CNF formula contains $\Theta(n^2)$ variables and $\Theta(n^3 + m^2t^2)$ clauses. Using a modern SAT solver, one
can evaluate small and medium size instances (with up to a few hundred of vertices) within a reasonable time.
For example, we computed optimal track layouts for all $977,\!526,\!957$ maximal planar graphs having $n = 18$ vertices. 
In total the computation took $5000$ machine-hours,
and all the graphs turned out to have a $t$-track layout for some $4 \le t \le 7$. Larger graphs, such as one in Corollary~\ref{cor:planar-track-8},
are solved within a few hours on a regular machine. Our implementation is available at~\cite{bob}.

\begin{figure}[!t]
	\begin{subfigure}[b]{.38\linewidth}
		\centering
		\includegraphics[page=1,width=1.0\textwidth]{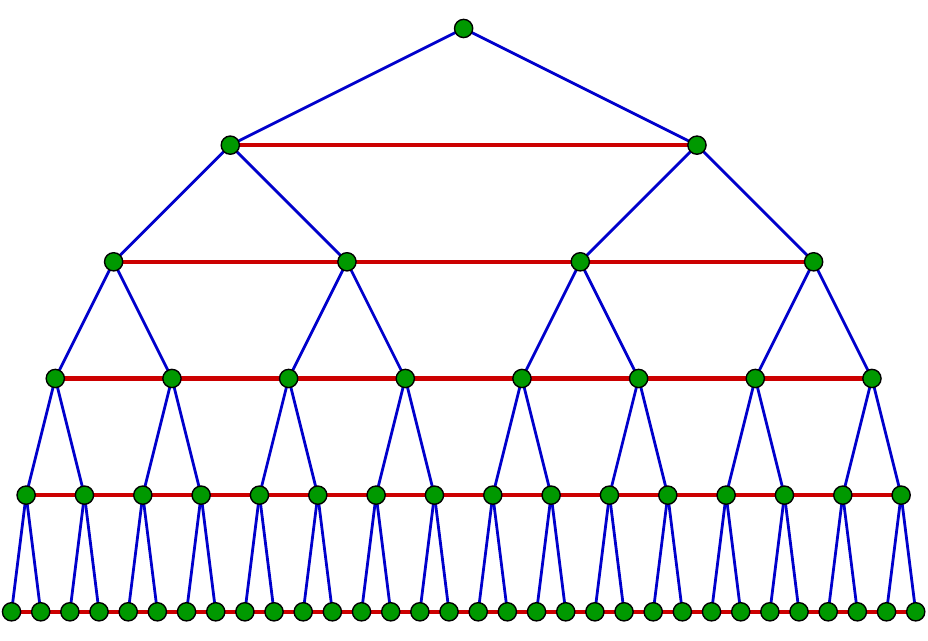}
		\caption{}
		\label{fig:lb_xtree}
	\end{subfigure}    
	\hfill
	\begin{subfigure}[b]{.38\linewidth}
		\centering
		\includegraphics[page=2,width=1.0\textwidth]{subclasses_lower_bounds}
		\caption{}
		\label{fig:lb_halin}
	\end{subfigure}
	\hfill
	\begin{subfigure}[b]{.2\linewidth}
		\centering
		\includegraphics[page=3,width=1.0\textwidth]{subclasses_lower_bounds}
		\caption{}
		\label{fig:weakly}
	\end{subfigure}    
	\caption{(a) An X-tree and (b) a Halin graph that require $5$ tracks.
		(c) A~weakly leveled planar graph that requires $6$ tracks; level edges are red.}
	\label{fig:lower}
\end{figure}

Using the formulation, we identified examples of an X-tree and a Halin graph that require $5$ tracks; see Figure~\ref{fig:lb_xtree} and \ref{fig:lb_halin}.
In particular, X-trees of depth $d \le 5$ admit a $4$-track layout but X-trees with $d \ge 6$ have track number $5$.
Similarly, we found a weakly leveled planar graph with $14$ vertices that has track number $6$; see Figure~\ref{fig:weakly}.
Thus, the algorithm of Bannister et al.~\cite{BDDEW15} for constructing $6$-track layouts of weakly leveled planar graphs is worst-case optimal. The results are summarized as follows.

\begin{theorem}
\label{thm:other}	
There exists an X-tree and a Halin graph with track number $5$.
There exists a weakly leveled planar graph with track number $6$.
\end{theorem}	

\section{Conclusions and Open Problems}
\label{sect:open}

In this paper we improved upper and lower bounds on the track number of several families of planar graphs.
A natural future direction is to close the remaining gaps for graphs listed in Table~\ref{table:tracks}. Next we discuss
several open questions related to track layouts.

Our approach for building track layouts of planar graphs relies on a construction for graphs of bounded treewidth. To the
best of our knowledge, the upper bound on the track number of $k$-trees is $(k + 1)(2^{k+1} - 2)^k$~\cite{Wie17}, while
the lower bound is only quadratic in~$k$~\cite{DMW05}. It seems unlikely that the existing upper bound is the right answer, and
finding a polynomial or even an exponential $2^{\Oh(k)}$ bound would already be an exciting improvement.

\begin{open}
	Improve the upper bound on the track number of $k$-trees.
\end{open}

One way of attacking the above problem is tightening a gap between track and queue numbers of a graph, as the
queue number of a $k$-tree is at most $2^k - 1$~\cite{Wie17}.
As mentioned earlier, every $t$-track graph has a $(t-1)$-queue layout. It is easy to see that this bound is worst-case optimal, that is, there exist $t$-track graphs that require $t-1$ queues in every layout. 
In the other direction, every $q$-queue graph has track number at most
$4q \cdot 4q^{(2q-1)(4q-1)}$~\cite{DPW04}. 
Similarly, every $q$-queue graph with acyclic chromatic number $c$ has track-number at most $c(2q)^{c-1}$~\cite{DMW05}.
Recall that a vertex coloring is \df{acyclic} if there is no bichromatic cycle.
These bounds can likely be improved. For example, $1$-queue graphs always admit a $4$-track layout
but no better bound is known for the case $q \ge 2$.

\begin{open}
	What is the largest track number of a $q$-queue graph?
\end{open}

Finally, we would like to see some progress on \emph{upward} track layouts~\cite{DLMW09,DW06}. An \df{upward track layout} of a 
dag $G$ is a track layout of the underlying undirected graph of $G$, such that the directed graph obtained from $G$ by adding
arcs between consecutive vertices in a track is acyclic. To the best of our knowledge, very little is known about this variant of layouts.
For example, the upward track number of (directed) paths and caterpillars
is known to be $3$~\cite{DLMW09}, while the upward track number of (directed) trees is at most~$5$~\cite{DW06}. Improving
the bounds and investigating other classes of graphs is an interesting research direction.

\begin{open}
	Investigate the upward track number of various families of graphs.
\end{open}

\section*{Acknowledgments}

We thank Jawaherul Alam, Michalis Bekos, Martin Gronemann, and Michael Kaufmann for fruitful initial discussions of the problem.

\bibliographystyle{abbrv}
\bibliography{refs}

\end{document}